\documentclass[letterpaper,UKenglish]{lipics}

\usepackage{microtype,amsmath, color}

\theoremstyle{plain}
\newtheorem{proposition}[theorem]{Proposition}
\newtheorem{question}[theorem]{Question}

\newcommand{\nc}{\newcommand}
\nc{\rnc}{\renewcommand}
\nc{\beq}{\begin{equation}}
\nc{\eeq}{{\end{equation}}}
\nc{\beqa}{\begin{eqnarray}}
\nc{\eeqa}{\end{eqnarray}}
\nc{\beeqa}{\begin{eqnarray*}}
\nc{\eeeqa}{\end{eqnarray*}}

\nc{\lbar}[1]{\overline{#1}}
\nc{\bra}[1]{\langle#1|}
\nc{\ket}[1]{|#1\rangle}
\nc{\ketbra}[2]{|#1\rangle\!\langle#2|}
\nc{\braket}[2]{\langle#1|#2\rangle}
\nc{\proj}[1]{| #1\rangle\!\langle #1 |}
\nc{\avg}[1]{\langle#1\rangle}
\nc{\Rank}{\operatorname{Rank}}
\nc{\smfrac}[2]{\mbox{$\frac{#1}{#2}$}}
\nc{\tr}{\operatorname{Tr}}
\nc{\ox}{\otimes}
\nc{\dg}{\dagger}
\nc{\dn}{\downarrow}
\nc{\cA}{{\cal A}}
\nc{\cB}{{\cal B}}
\nc{\cC}{{\cal C}}
\nc{\cD}{{\cal D}}
\nc{\cE}{{\cal E}}
\nc{\cF}{{\cal F}}
\nc{\cG}{{\cal G}}
\nc{\cH}{{\cal H}}
\nc{\cI}{{\cal I}}
\nc{\cJ}{{\cal J}}
\nc{\cK}{{\cal K}}
\nc{\cL}{{\cal L}}
\nc{\cM}{{\cal M}}
\nc{\cN}{{\cal N}}
\nc{\cO}{{\cal O}}
\nc{\cP}{{\cal P}}
\nc{\cR}{{\cal R}}
\nc{\cS}{{\cal S}}
\nc{\cT}{{\cal T}}
\nc{\cX}{{\cal X}}
\nc{\cY}{{\cal Y}}
\nc{\cZ}{{\cal Z}}
\nc{\csupp}{{\operatorname{csupp}}}
\nc{\qsupp}{{\operatorname{qsupp}}}
\nc{\var}{{\operatorname{var}}}
\nc{\rar}{\rightarrow}
\nc{\lrar}{\longrightarrow}
\nc{\polylog}{{\operatorname{polylog}}}
\nc{\1}{{\leavevmode\hbox{\small1\normalsize\kern-.33em1}}}
\nc{\wt}{{\operatorname{wt}}}
\nc{\av}[1]{{\left\langle {#1} \right\rangle}}

\nc{\FF}{{{\mathbb F}}}
\nc{\NN}{{{\mathbb N}}}
\nc{\ZZ}{{{\mathbb Z}}}
\nc{\PP}{{{\mathbb P}}}
\nc{\QQ}{{{\mathbb Q}}}
\nc{\UU}{{{\mathbb U}}}
\nc{\EE}{{{\mathbb E}}}
\nc{\id}{{\operatorname{id}}}

\nc{\CHSH}{{\operatorname{CHSH}}}

\nc{\be}{\begin{equation}}
\nc{\ee}{{\end{equation}}}
\nc{\bea}{\begin{eqnarray}}
\nc{\eea}{\end{eqnarray}}
\nc{\<}{\langle}
\rnc{\>}{\rangle}
\nc{\Hom}[2]{\mbox{Hom}( {\bf C}^{#1}, {\bf C}^{#2})}
\nc{\rU}{\mbox{U}}

\nc{\ob}[1]{#1}

\def\mat{Mat\'{u}\v{s} }

\def\ovb{\overline}
\def\ot{\otimes}
\def\tr{{\rm Tr} \, }
\def\trp{{\rm Tr} }

\def\half{\tfrac{1}{2}}
\def\wtd{\widetilde}

 \def\ds{\displaystyle}

\def\ing{{\rm Ing}}
\def\nn{\nonumber}
\def\wh{\widehat}

\title{The Quantum Entropy Cone of Stabiliser States}  
\titlerunning{Entropy Inequalities for Stabiliser states} 

\author[1]{Noah Linden}
\author[2]{Franti\v{s}ek Mat\'u\v{s}}    
\author[3,4]{Mary Beth Ruskai}
\author[5,1,6]{Andreas Winter}
\affil[1]{School of Mathematics, University of Bristol  \\ Bristol BS8 1TW, United Kingdom \\
  \texttt{n.linden@bristol.ac.uk}}
  \affil[2]{Institute of Information Theory and Automation \\ Academy of Sciences of the Czech Republic  \\ 
        Prague, Czech Republic\\
  \texttt{matus@utia.cas.cz}}
  \affil[3]{Institute for Quantum Computing, University of Waterloo\\
  Waterloo, Ontario, Canada\\
  \texttt{mbruskai@gmail.com} }
      \affil[4]{Tufts University, Medford, MA 02155 USA \\  \texttt{marybeth.ruskai@tufts.edu}}
  \affil[5]{ICREA \&{} F\'{\i}sica Te\`{o}rica: Informaci\'{o} i Fenomens Qu\`{a}ntics \\ Universitat Aut\`{o}noma de Barcelona\\ ES-08193 Bellaterra (Barcelona), Spain \\
    \texttt{andreas.winter@uab.cat}}
    \affil[6]{Centre for Quantum Technologies, National University of Singapore \\
 2 Science Drive 3, Singapore 117542, Singapore}

\date{24 June 2013}

\authorrunning{Linden,  Mat\'u\v{s},  Ruskai and Winter} 

\Copyright[by]{Linden, Mat\'u\v{s},  Ruskai and Winter}


\serieslogo{}
\volumeinfo
  {Billy Editor, Bill Editors}
  {2}
  {Conference title on which this volume is based on}
  {1}
  {1}
  {1}
\EventShortName{}
\DOI{10.4230/LIPIcs.xxx.yyy.p}

\begin{document}

\maketitle

\begin{abstract}
We investigate the universal linear inequalities that hold for the
von Neumann entropies in a multi-party system, prepared in a stabiliser state.
We demonstrate here that entropy vectors for stabiliser states
satisfy, in addition to the classic inequalities,
 a type of linear rank inequalities associated with  the 
combinatorial structure of normal  subgroups of certain matrix groups.

 In the 4-party case, there is only one such inequality, the
so-called Ingleton inequality.
For  these systems we  show that strong
subadditivity, weak monotonicity and Ingleton inequality exactly 
characterize the entropy cone for stabiliser states.
 \end{abstract}

\section{Introduction}    \label{sect:Intro}

\subsection{Background}   \label{sect:backgrnd}

Undoubtedly, the single most important quantity in (classical) information
theory is the Shannon entropy, and its properties play a key role: for
a discrete probability distribution $p$ on $\cT$
\begin{align}      \label{eqdef:H}
   H(p) =    - \sum_{ t \in \cT } p(t)\log p(t) ~.   \end{align}
 The quantum (von Neumann) entropy is understood to be 
of equal importance to quantum information: for a quantum state
(density operator) $\rho \geq 0$, $\tr\rho =1$
\beqa  \label{eqdef:S}
   S(\rho) = -\tr\rho\log\rho 
\eeqa
which reduces to \eqref{eqdef:H} when $\rho$ is diagonal.

For $N$-party systems, one can apply these definitions to obtain
the entropy of all marginal probability distributions (in the classical case)
and reduced density operators (aka quantum marginals) in the quantum
case.   The collection of these entropies can be regarded as a vector
in ${\bf R}_{2^N}$, and the collection of all such vectors forms a set
whose closure is a  convex  cone.   It is an interesting open question
to determine the inequalities which characterize this cone.   As discussed in
Section~\ref{nonShan},
  it is now known   that in the  classical setting the Shannon inequalities
given below do not suffice; they describe a strictly larger cone.

This work has motivated us to consider analogous questions for the
von Neumann entropy in $N$-party quantum systems.   Although we
are unable to answer this question, we can fully characterize
the cone associated with a subset of quantum states known
as stabiliser states in the 4-party case.   Moreover, we can 
show that for any number of parties, entropy vectors for
stabiliser states satisfy additional inequalities in the
class known as linear rank inequalities discussed 
in Section~\ref{sect:linrank}.   In the classical setting, distributions
whose entropies satisfy this subclass of stronger inequalities, suffice
to achieve maximum throughput in certain network coding problems \cite{LYC}.

\subsection{Classic inequalities and Definitions}   \label{sect:classic}

It is well-known that the classical Shannon entropy for an
$N$-party classical probability distribution $p$ on a discrete
space $\cT_1\times\cdots\times\cT_N$, has the following properties,
commonly known as the \emph{Shannon inequalities}:
\begin{enumerate}
  \item It is non-negative, \emph{i.e.} $H(A) \geq 0$; $H(\emptyset)=0$. \hfill (+)

  \item It is strongly subadditive (aka submodular), \emph{i.e.}
    \begin{equation}
      \label{ssa}
      H(A) + H(B) - H(A\cap B) - H(A\cup B) \geq 0.  \tag{SSA}
    \end{equation}

  \item It is monotone non-decreasing, \emph{i.e.}
    \begin{equation}
      \label{mono}
      A \subset B \quad \Longrightarrow \quad H(A) \leq H(B).  \tag{MO}
    \end{equation}
\end{enumerate}
where $H(A)$ denotes the  entropy 
$H(p_A)$ of the marginal distribution $p_A$ on 
$\cT_A = \bigotimes_{j\in A} \cT_j$.

 
The monotonicity property \eqref{mono} implies that if   
$H(A) = 0 $  then $H(B) = 0$ for all $B \subset A$
and, thus, $p_A = \bigotimes_{j\in A} \delta_{t_j}$ is a product of
point masses. 
Some of the most remarkable features of quantum
systems arise when \eqref{mono} is violated.
Indeed, for a pure entangled state $\rho_{AB}  = \proj{\psi}_{AB}$
for which $ S(\rho_{AB}) = 0$, but the entropy of the reduced states  $\rho_A = \trp_B  \, \rho_{AB}$
and $ \rho_B = \trp_A \, \rho_{AB}$   can be (and usually is) strictly positive.
In fact, $S(\rho_{AB})  - S(\rho_A)$ can be as large as $-\log d$, 
where $d$ is the Hilbert space dimension of the smaller of $A$ and $B$.

For multi-party quantum systems, (+) and \eqref{ssa} are still
valid~\cite{LiebRuskai}, but \eqref{mono} has to be replaced by the 
third property below -- in analogy to the classical case, we call them
Shannon inequalities:

\begin{enumerate}
  \item Non-negativity: $S(A) \geq 0$; $S(\emptyset)=0$. \hfill (+)

  \item Strong subadditivity:
    \begin{equation}
      \label{eq-ssa}
      S(A) + S(B) - S(A\cap B) - S(A\cup B) \geq 0.  \tag{SSA}
    \end{equation}

  \item Weak monotonicity:
    \begin{equation}
      \label{eq:WMO}
      S(A) + S(B) - S(A\setminus B) - S(B\setminus A) \geq 0. \tag{WMO}
    \end{equation}
\end{enumerate}
However, in contrast to the classical setting, this weaker
version of monotonicity is not completely independent
of strong subadditivity \eqref{ssa}. In fact, it can be obtained
from the latter by the (non-linear) process known as \emph{purification} 
described in Section~\ref{sect:pure}.
Using a slight abuse of notation, we use $I(A:B) $ and $I(A:B|C) $ to denote,
respectively,
the mutual information and conditional mutual information for both classical
and quantum systems, defined explicitly  in the latter case as
\begin{align*}
   I(A:B)   & = S(A) + S(B) - S({AB}) , \\
   I(A:B|C) & = S({AC}) + S({BC}) - S({ABC}) - S(C) ,
\end{align*}
for pairwise disjoint sets $A$, $B$, $C$.
Note that SSA can then be written as $I(A:B|C) \geq 0$.

\subsection{Entropy cones and non-Shannon inequalities}   \label{nonShan}

The first non-Shannon entropy inequality was obtained in 1997-98 by
Yeung and Zhang~\cite{Yeung:framework,ZhangYeung-1,ZhangYeung-2} 
for $4$-party systems.
Their work established that the classical entropy  cone 
is strictly smaller than the polyhedral cone defined by the Shannon  inequalities.
This was  the only non-Shannon inequality  known until 2006, when
Dougherty, Freiling and Zeger~\cite{DFZ2006,DFZ2007}
used a computer search to generate new inequalities.
Then Mat\'{u}\v{s}~\cite{Matus} found two infinite families,  one of which can be written as 
\beqa   \label{eq:Matus-s}
    t \,   \ing(AB:CD)   +  I(A:B|D)    + \frac{t (t +1)}{2} \bigl[ I(B:D|C) + I(C: D|B) ) \big]  & \geq  &0  
 \eeqa
where  $t$ is a non-negative integer, and   $ \ing(AB:CD) $ is defined in \eqref{eq:ingleton} below.
The case $t = 1$ in  \eqref{eq:Matus-s} yields the  inequality in \cite{ZhangYeung-2}.
Moreover, either of the Mat\'{u}\v{s} families can be used  to show that the 4-party
entropy cone is not polyhedral.
In     \cite{DFZ2011}  additional non-Shannon inequalities were found.   

In the quantum setting, Lieb~\cite{Lieb-inequalities} considered the question
of additional inequalities in a form that could be regarded as extending  SSA 
to more parties, but found none.    Much later Pippenger~\cite{Pippenger} rediscovered
one of Lieb's results and used it to show constructively that there are no additional
inequalities for 3-party systems.
He also explicitly raised the question of whether or not additional inequalities
hold for more parties.   Despite the fact that \eqref{eq-ssa} is still the only 
known inequality, it has been shown that for 4-party systems
there are constrained inequalities 
 \cite{CLW, LindenWinter:new} that do not follow from SSA.
  (Numerical evidence for additional inequalities  is
  given in the thesis of Ibinson \cite{Ibinson-PhD}.)

\subsection{Structure of the paper}    \label{sect:struc}

This paper is organized as follows.   In Section~\ref{sect:prelim}
we give some basic notation and review some well-known facts.
  In Section~\ref{sect:linrank} we
discuss what is known about linear rank inequalities 
beginning with the Ingleton inequality in Section~\ref{sect:ing}
and concluding with a discussion of their connection to
the notion of common information in Section~\ref{comminfo}.
In Section~\ref{sect:stab} we discuss stabiliser states,
beginning with some basic definitions in Section~\ref{stabdefs}.
In Section~\ref{sect:stab.ent} we consider the entropies of
stabiliser states, showing half of our main result that pure 
stabiliser states generate entropy vectors which satisfy the
Ingleton inequality and a large class of other linear rank inequalities.
In Section~\ref{sect:4party} we prove the other half, \emph{i.e.}, 
that all extremal rays of the 4-party Ingleton cone can be
achieved using 5-party stabiliser states.   We conclude with
some open questions and challenges.

\section{Preliminaries}  \label{sect:prelim}

\subsection{Notation}  \label{sect:notation}

We now introduce some notation needed to make precise
the notion of entropy vectors and entropy cones.
We will let $\cX = \{A, B, C, \ldots \}$ denote an index set
of finite size $|\cX|=N$ so that in many cases we could just
assume that $\cX =  \{1, 2, \ldots N\}$. 
However, it will occasionally be useful to consider
the partition  of some the index set into smaller groups,
e.g, by grouping $A$ and $B$ as well as $D$ and $E$,
$\cX_5 = \{ A,B,C,D,E \}$ gives rise to a $3$-element
$\cX_3 = \{ AB,C,DE \}$. 
When the size of $\cX$ is important, we write $\cX_N$.  

An arbitrary $N$-partite
quantum system is associated with a Hilbert space
$ \cH = \bigotimes_{x\in{\cal X}} \cH_x$ 
(with no restrictions on the dimension  of the Hilbert spaces $\cH_x$)
with $|{\cal X}|=N$.
The  reduced states (properly called reduced density operators, but more
often referred to as  reduced density matrices (RDM) and also
known as quantum marginals) 
are given by $\rho_J = \tr_{J^c} \rho$, where
$J^c = {\cal X}\setminus J$. This gives rise to a function
$   S: J \longmapsto S(J)= S(\rho_J)   $ on the subsets $J \subset {\cal X}$.  
An element of the output of $S$ can be viewed as a vector in ${\bf R}^{2^{\cX}}$,
whose coordinates are indexed by the power set $2^{\cX}$ of $\cX$.
We study the question of which such vectors arise from classical or
quantum states, \emph{i.e.}, when their elements are given by the entropies
$S(\rho_J)$ of the reduced states of some $N$-party quantum state.

Classical probability distributions can be embedded into the
quantum framework by restricting density matrices to those
which are diagonal in a fixed product basis.
A function $H:2^\cX \rightarrow {\bf R}$,
associating real numbers to the subsets of a finite set $\cX$, which 
satisfies the Shannon inequalities, eqs.~(+), \eqref{ssa} and \eqref{mono},
is called \emph{poly-matroid}. 
By analogy with poly-matroids, we propose to call a function
$S:2^\cX \rightarrow {\bf R}$ a \emph{poly-quantoid}, if it satisfies (+), \eqref{ssa}
and \eqref{eq:WMO}~\cite{M.poly}. 

  We will let  $\Gamma_\cX^C$ and $\Gamma_\cX^Q$ denote, respectively, the
  convex cone of vectors in a poly-matroid or poly-quantoid.   The existence of non-Shannon
  entropy inequalities implies that there are vectors in $\Gamma_\cX^C$
  which can not be achieved by any classical state.    Neither the classical nor quantum
   set of true entropy
  vectors is  convex, because their boundaries have a complicated structure
  \cite{CLW,LindenWinter:new,Mbd,Pippenger}.
    However, the closure of the set of classical or quantum entropy vectors,
   which we denote $\ovb{\Sigma}^C_\cX $ or    $\ovb{\Sigma}^Q_\cX$,    respectively,
  is a closed convex cone.  The inclusion $ \ovb{\Sigma}^C_\cX   \subset    \Gamma_\cX^C$ 
   is strict for $N \geq 4$~\cite{ZhangYeung-2}. 
   It is an important open question whether or not this also holds in
   the quantum setting, \emph{i.e.}, is the inclusion
   $ \ovb{\Sigma}^Q_\cX   \subseteq    \Gamma_\cX^Q$ also  strict?
   
In this paper, we consider entropy vectors which satisfy additional
inequalities known as linear rank inequalities, i.e. those satisfied by
the dimensions of subspaces of a vector space and their intersections.
A poly-matroid $H$ is called \emph{linearly represented} if $H(J) = \dim \sum_{j\in J} V_j$
for subspaces $V_j$ of a common vector space $V$. 
  
The simplest linear rank inequality is the $4$-party Ingleton inequality
(see section~\ref{sect:linrank} below).
Poly-matroids and poly-quantoids which
also satisfy these additional inequalities will be denoted
$\Lambda_\cX^C$ and $\Lambda_\cX^Q$ respectively.

\subsection {Purification and complementarity}  \label{sect:pure}

For statements about $J$ and $J^c = \cX \setminus J$, it suffices
to consider a bipartite  quantum system with Hilbert
spaces $\cH_A$ and $\cH_B$.    It is well-known that
any pure state $\ket{\psi_{AB}}$ can be written in the
form
\beqa
   \ket{\psi_{AB}} = \sum_k \mu_k \ket{\phi_k^A} \ot \ket{\phi_k^B}
\eeqa
with $\mu_k > 0$ and $\{ \phi_k^A \}$ and $\{ \phi_k^B \}$
orthonormal.   Indeed, this is an immediate consequence 
of the isomorphism between  $\cH_A \ot \cH_B$ and 
$\cL(\cH_A,\cH_B)$, the set of linear operators from
$\cH_A$ to $\cH_B$, and the singular value decomposition.
It then follows that both $\rho_A$ and $\rho_B$ have the same
non-zero eigenvalues $\mu_k^2$, and hence $S(\rho_A) = S(\rho_B)$.

This motivates the process known as {\rm purification}.
Given a density matrix $\rho = \sum_k \lambda_k \proj{\phi_k}$,
one can define the bipartite  state
$$\ket{\psi} = \sum_k  \sqrt{\lambda_k} \,  \ket{\phi_k}  \ot  \ket{\phi_k }$$
whose reduced density matrix  $\trp_B  \, \proj{\psi} $ is $\rho$.

Therefore, every vector in an $N$-party quantum entropy cone 
$ \Sigma^Q_N $ can be obtained from the entropies of some
  reduced state of a  $(N+1)$-party pure state $\ket{\Psi}$.  In
  that case, we say that the entropy vector is realized by  $\ket{\Psi}$.

  In an abstract setting, we could define a cone   $\wtd{\Gamma}_\cX^Q$
   whose elements satisfy ($+$), \eqref{eq-ssa} and the complementarity
   property  $S(J)= S(J^c)$, and let $\Gamma_N^Q$ be the cone of vectors
  which arise as subvectors of $ \wtd{\Gamma}_{N+1}^Q$.    Although we
  will not need this level of abstraction, this correspondence is used in Section~\ref{sect:4party}.

\subsection{Group inequalities} \label{sect:group}
   
 Consider a (finite) group $G$ and a family of
subgroups $G_x < G$, $x\in{\cal X}$. Then,
$H(J ) = \log|G/G_J|$, with $G_J = \bigcap_{j\in J } G_j$
is a poly-matroid. In fact, Chan and Yeung~\cite{ChanYeung-group}
show that it is entropic because it can be realised by
the random variables $X_j = gG_j \in G/G_j$ for a uniformly
distributed $g\in G$.
The fact that for two subgroups $G_1,G_2 < G$, the mappings
\begin{align*}
  G/(G_1\cap G_2) \longrightarrow G/G_1 \times G/G_2  \qquad \text{and} \qquad 
  g(G_1\cap G_2)  \longmapsto     (gG_1,gG_2),
\end{align*}
are one-to-one~\cite{Suzuki:groups}, guarantees that indeed $H(X_J) = H(J)$.

Thus, the inequalities satisfied by poly-matroids, and more specifically
entropic poly-matroids give rise to relations between the cardinalities of
subgroups and their intersections in a generic group. Conversely,
Chan and Yeung \cite{ChanYeung-group} have shown that every such relation for groups,
is valid for all entropic poly-matroids.
This result motivates the search for a similar, combinatorial or group
theoretical, characterization of the von Neumann entropic poly-quantoids,
and our interest in stabiliser states originally grew out of it.

However, it must be noted that if some subgroups of $G$ are not general,
but, e..g,  normal subgroups as in Theorem~\ref{thm:normsubgp} below,
 then the Chan-Yeung correspondence
breaks down. In this case further inequalities
hold for the group poly-matroid that are not satisfied by entropic 
poly-matroids.

\section{Linear rank inequalities}  \label{sect:linrank}

\subsection{The Ingleton inequality}   \label{sect:ing}

The classic \emph{Ingleton inequality}, 
when stated in information theoretical
terms, and as manifestly balanced, reads
\begin{equation}
    \ing(AB : CD ) \equiv  I(A :B|C) + I(A:B|D) + I(C:D) - I(A:B) \geq 0,  \tag{ING}
  \label{eq:ingleton}
\end{equation}
where $A$, $B$, $C$ and $D$ are elements (more generally pairwise disjoint
subsets) of ${\cal X}$. It was discovered by Ingleton~\cite{Ingleton} as a constraint
on linearly represented matroids.

Although this inequality does not hold universally, it is of 
considerable importance, and continues to be studied \cite{MS,MTH,SO,M.poly},
particularly when reformulated as an inequality for subgroup ranks.
In Theorem~\ref{thm:stab-ingleton} we show that
\eqref{eq:ingleton} always holds for a special class of
states.   Before doing that, we give some basic properties
first. Observe that \eqref{eq:ingleton} is symmetric with respect to
the interchanges  $A \leftrightarrow B $ and $C \leftrightarrow D$,
so that it suffices to consider special properties only for $A$ and $D$.

Because it is not always easy to see if a 4-party state  
$\rho_{ABCD}$ is the reduction of a pure stabiliser state,
it is worth listing some easily checked conditions 
under which \eqref{eq:ingleton} holds.

\begin{proposition}  \label{thm:Ingprop}
The Ingleton inequality \eqref{eq:ingleton} 
holds if any one of the following conditions holds.

a) $\rho_{ABCD} = \proj{\psi_{ABCD}}$ is any pure 4-party state.

b) $\rho_{ABCD} = \rho_{ABC} \ot \rho_D$ or $\rho_{A} \ot \rho_{BCD}$

c)  The two-party component of the entropy vector for $ (\rho_{ABCD} ) $ is 
symmetric under a partial exchange between $(A,B) $ and $ (C,D)$, \emph{i.e.}
under any {\em one} (but not two) of the exchanges  $A \leftrightarrow C$,
$B \leftrightarrow D$, $A \leftrightarrow D$ or $B \leftrightarrow C$.

\end{proposition}

\begin{proof} To prove (a) it suffices to observe that
\beeqa   \label{eq:pureIng}
 \ing(AB : CD) & = &   I(A:B|C) 
         + S(AD) + S(BD) - S(D) - S(ABD)  \\
 & ~ & \quad  + \, S(C) + S(D) - S(CD) 
   - S(A) - S(B) + S(AB) \\
& = & I(A:B|C)   
    + S(AD) + S(AC) - S(A) - S(ACD)  \\
    & = &  I(A:B|C) +   I(C:D | A)  \geq 0 .
\eeeqa
To prove (b)  observe that when  $ \rho_{ABCD} = \rho_{ABC} \ot \rho_D$ then
 $ I(A:B|D) = I(A:B) $ and \linebreak $I(C:D) = 0$ so that  \eqref{eq:ingleton}  follows immediately
 from   \eqref{ssa}.   For $ \rho_{ABCD} = \rho_{A} \ot \rho_{BCD}$
 the first, second and last terms in \eqref{eq:ingleton}  are zero 
 so that it becomes simply $I(C:D) \geq 0$.
   
For (c)  we observe that 
\eqref{eq:ingleton} is equivalent to
\beqa   
  I(B:C |A) + I(A:D |B) + R \geq 0 \quad \hbox{with} \quad   
  R = S(BC) + S(A D) - S(C D) - S(A B) .
\eeqa
The exchange  $A \leftrightarrow C$ takes $R$ to $-R$.
Thus, if   $\rho_{ABCD} $ is symmetric under this exchange,
 then  $R = 0$ and \eqref{eq:ingleton} holds.    The sufficiency of the other
 exchanges can be shown similarly.
 \end{proof}

If (\ref{eq:ingleton}) holds, then all of the
Mat\'{u}\v{s}   inequalities \eqref{eq:Matus-s} also hold,
since they add only conditional mutual informations $I(X:Y|Z) \geq 0$ to it.
However, it is   well-known that entropies do not universally obey
the Ingleton inequality. A simple, well-known counterexample is given by
independent and uniform binary variables $C$ and $D$, and
$A = C \vee D$, $B = C \wedge D$.
Then the first three terms in \eqref{eq:ingleton} vanish, so
that $ \ing(AB : CD) = - I(A:B) < 0$.
  
To obtain a quantum state which violates Ingleton, let
$ \ket{\psi} = \tfrac{1}{\sqrt{2}} \big( \ket{0000}  + \ket{ 1111 }\big) $ and
\beqa  \label{quant_viol_Jng}
 \rho_{ABCD} = \half \proj{\psi} + \tfrac{1}{4} \proj{1010}  + \tfrac{1}{4} \proj{1001} .
 \eeqa
All  the reduced states  $\rho_{ABC}, \rho_{BD} $, etc. are separable
and identical to those of the   state
\beeqa
 \rho_{ABCD} & = &  \tfrac{1}{4} \proj{0000}  + \tfrac{1}{4} \proj{1111}  
     + \tfrac{1}{4} \proj{1010 }  + \tfrac{1}{4} \proj{1001 } .
  \eeeqa
corresponding to the classical example above.   Therefore \eqref{quant_viol_Jng}
violates the Ingleton inequality, but still satisfies all of the Mat\'{u}\v{s}
inequalities.   Note that the state  $\ket{\psi} $ is maximally entangled
wrt the splitting $A$ and $BCD$.
Additional quantum states with the same entropy vectors as
classical states which violate Ingleton \cite{MS,M2} can be similarly constructed.
However, we do not seem to know ``genuinely quantum'' counterexamples
to the Ingleton inequality.

\begin{question}
Do there exist quantum states which violate Ingleton and are neither
separable nor have the same entropy vectors as some classical state?
\end{question}

\subsection{Families of inequalities}   \label{sect:kinser}

When the subsystem $C$ or $D$ is trivial, the Ingleton
inequality reduces to the $3$-party SSA inequality,
$I(A:B |C) \geq 0 $;
when subsystem $A$ or $B$  is trivial, it reduces
to the $2$-party subadditivity inequality $I(C:D) \geq 0$.
This suggests that the Ingleton inequality is a member
of a more general family of $N$-party inequalities.
In 2011, Kinser \cite{Kinser} found the first such family, 
which can be written (for $N \geq 4$) as
\begin{align}  \label{eq:kinser}  
 K[N]  & =   I(1:N|3) + H(1N) - H(12) -H(3N) +H(23)  
     +   \sum_{k = 4}^{N} I(2:k-1 | k) \geq 0. 
\end{align}
This is equivalent to the Ingleton inequality when $N = 4$.

\begin{remark}
As in the proof of Proposition~\ref{thm:Ingprop}(c), 
it can be shown  that  Kinser's inequalities
hold if $\rho$ is symmetric with respect to the interchange $1 \leftrightarrow 3$
or $2 \leftrightarrow N$. They also hold if
$\rho_{1,2, \ldots N} = \rho_2 \ot \rho_{1,3, \ldots N} $
One can ask if   part (a) of Theorem~\ref{thm:Ingprop} 
can be extended to the new inequalities, \emph{i.e.},  do they hold
for  hold for $N$-party pure quantum states? 
\end{remark}

\subsection{Inequalities from common information}   \label{comminfo}

Soon after Kinser's work, another group   \cite{DFZ0910}
found new families of linear rank inequalities for poly-matroids
for $N > 4$ that are independent of both Ingleton's
inequality and Kinser's family.   In the $5$-party case,
they found a set of $24$ inequalities which generate
 all linear rank inequalities for poly-matroids.  
 Moreover, they give an algorithm
 which allows one to generate many more families of
 linear rank inequalities based on the notion of common information,
 considered  much earlier in  \cite{AG,AK,GK} and used below. 
However, it was shown in \cite{CGK2} that there are $N$-party linear 
rank inequalities that cannot be obtained from the process 
described in   \cite{DFZ0910}.
 
 \begin{definition}
  \label{defi:common-info}
  In a poly-matroid $H$ on ${\cal X}$, two subsets $A$ and $B$
  are said to \emph{have a common information}, if there
  exists an extension of $H$ to a poly-matroid on the larger set
  ${\cal X} \stackrel{.}{\cup} \{\zeta\}$, such that
$    H(\{\zeta\}\cup A) - H(A) =: H(\zeta|A)  = 0$, 
 $   H(\{\zeta\}\cup B) - H(B) =: H(\zeta|B)  = 0$
   and $H(\zeta)           = I(A:B)$.
\end{definition}

Here we used $H( Z | A) = H(AZ)  - H(A) $ to denote the conditional entropy.
For completeness we include a proof (courtesy of a Banff talk by Dougherty) 
of the following result,  
as well as a proof of  Lemma~\ref{lemma:mystery} below, which appear in \cite{DFZ0910}.
\begin{proposition}
  \label{prop:c-i-implies-ING}
  Let $h$ be a poly-matroid on ${\cal X}$, and $A,B,C,D \subset {\cal X}$
  such that $A$ and $B$ have a common information. Then
  the Ingleton inequality (\ref{eq:ingleton}) holds for $A$,
  $B$, $C$ and $D$.
\end{proposition}
\begin{proof}
Let $\zeta$ be a common information of $A$ and $B$. Then, using
$H(F|A) \geq H(F|AC) $ in Lemma~\ref{lemma:mystery} below,
and letting $F = \zeta $, gives 
$$    I(A:B|C) + H(\zeta|A) \geq I(\zeta:B|C)   .$$
Using this a total of six times, we obtain
\[\begin{split}
  I(A:B|C) &    + I(A:B|D) + I(C:D) + 2H(\zeta|A) +2H(\zeta|B) \\
           &\geq  I(A:\zeta|C) + I(A:\zeta|D) + I(C:D) + 2H(\zeta|A) \\
           &\geq  I(\zeta:\zeta|C) + I(\zeta:\zeta|D) + I(C:D) \\
           &=     H(\zeta|C) + H(\zeta|D) + I(C:D)  
             \geq H(\zeta|C) + I(\zeta:D)
            \geq  I(\zeta:\zeta) = H(\zeta).
\end{split}\]
Inserting the conditions for $\zeta$ being a common information,
completes the proof.
\end{proof}

\begin{lemma}
  \label{lemma:mystery}
  In a poly-matroid $H$ on a set ${\cal X}$  with subsets $A,B,C,F \subset {\cal X}$.
\begin{align}   \label{eq:myst}
 I(A:B|C) + H(F|AC)  \geq I(F:B|C)  \end{align}
 \end{lemma}
\begin{proof}
 By  a direct application of the poly-matroid axioms:
\begin{align}
 \lefteqn{ I(A:B|C) + H(F|A  C) - I(F:B|C) 
           =    H(B|F  C) - H(B|A  C) + H(F|A  C)} \quad \nn  \\
           &=    H(B  C  F) + H(A  C  F)    
                   - H(C  F) - H(A  B  C)               \\
           &\geq H(B  C  F) + H(A  C  F)    
                  - H(C  F) - H(A  B  C  F) \nn  \\
           &=    I(A:B|C  F) \geq 0,
\end{align}
where we used only algebraic identities, SSA and monotonicity.
\end{proof}

In a linearly represented poly-matroid, (\ref{eq:ingleton})
is universally true: There, $H(J)  = \dim V_J$, with
$V_J = \sum_{j\in J} V_j$ for a family of linear subspaces
$V_j  \subset V$ of a vector space. The common information of any
$A,B \subset {\cal X}$ is constructed by defining
$V_\zeta = V_A \cap V_B$.

\begin{theorem}   \label{thm:normsubgp}
Any linear rank inequality for a poly-matroid obtained using common
 information and the poly-matroid inequalities, also holds for a group
 poly-matroid when its defining subgroups are normal. 
\end{theorem}
\begin{proof}  
It suffices to show that when
 $G_A, G_B \vartriangleleft G$ are normal subgroups
for $A,B\subset {\cal X}$, then $A$ and $B$ have a common
information given by $G_\zeta = G_A G_B \vartriangleleft G$
(the latter from the normality of $G_A$ and $G_B$).
The first two conditions for a common information are
clearly satisfied, as $G_A, G_B \subset G_A G_B$, and
the third follows from the well-known natural isomorphisms 
\cite{Suzuki:groups}
\begin{align*}
  G/(G_A G_B)    \eqsim \raisebox{2pt}{$(G/G_A)$}{\Big/}\raisebox{-2pt}{\!$\bigl((G_A G_B)/G_A\bigr)$} \quad
                   \text{ and}   \quad
  (G_A G_B)/G_A &\eqsim G_B/(G_A\cap G_B),
\end{align*}
which imply
\begin{align*}
  H(\zeta) &= \log|G/(G_AG_B)|  
            = \log|G/G_A| - \log|(G_A G_B)/G_A| \\
           &= \log|G/G_A| + \log|G/G_B| - \log|G/(G_A\cap G_B)| 
            = I(A:B).    \qedhere
\end{align*}
\end{proof}

\section{Entropies of stabiliser states}   \label{sect:stab}

\subsection{Stabiliser groups and stabiliser states}    \label{stabdefs}

Motivated by the stabiliser states encountered in the extremal
rays of $\Sigma_2$, $\Sigma_3$ and $\Sigma_4$, we now focus on (pure) 
stabiliser states, \emph{i.e.}~$1$-dimensional quantum codes. 
Stabilizer codes have 
emerged in successively more general forms. We use the
formulation described by Klappenecker and 
R\"otteler~\cite{KlappiRoetti:nice-error-1,KlappiRoetti:nice-error-2}
(following Knill~\cite{Knill:error-groups})
which relies on the notion of \emph{abstract error group}:
This is a finite subgroup $W < {\cal U}(\cH)$ of the unitary
group of a (finite dimensional) Hilbert space $\cH$,
which satisfies the following axioms:
\begin{enumerate}
  \item The center $C$ of $W$ consists only of multiples
        of the identity matrix (``scalars''): $C \subset  {\bf C}\1$.
  \item $\wh{W} \equiv W/C$ is an abelian group of order $|\cH|^2$, called the
        \emph{abelian part} of $W$.
  \item For all $g\in W\setminus C$, $\tr g = 0$.
\end{enumerate}
Note that conditions 1 and 2 imply that $W$ is non-abelian; whereas
condition 2 says that the non-commutativity is played out only
on the level of complex phases: for $g,h \in W$,
\[
  gh = \omega(g,h)hg, \text{ with } \omega(g,h) \in  {\bf C}.
\]
Finally, condition 3 means that $g,h \in W$ in
different cosets modulo $C$ are orthogonal with respect to
the Frobenius (or Hilbert-Schmidt) inner product:
$\tr g^\dagger h = 0$.
It is known that $\wh{W}$ is a direct product of an abelian group $T$ with itself, 
such that $|T| = |\cH|$. 

\begin{example}[Discrete Weyl-Heisenberg group]
  \label{ex:weyl-heisenberg}
  \normalfont
  Let $\cH$ be a $d$-dimensional Hilbert space, 
  with a computational orthonormal basis
  $\{\ket{j}\}_{j=0}^d$. Define discrete Weyl operators
  \begin{align*}
    X\ket{j}  = \ket{j+1} \!\mod\! d , \qquad  \qquad
    Z\ket{j}  = e^{j\frac{2\pi i}{d}}\ket{j}.
  \end{align*}
  They are clearly both of order $d$, and congruent via the
  discrete Fourier transform.
    The fundamental commutation relation,
 $  XZ = e^{2\pi i/d} ZX $
ensures that the group $W$ generated by $X$ and $Z$ is finite,
  and indeed an abstract error group with center 
  $C = \left\{ e^{j\frac{2\pi i}{d}} : j=0,\ldots,d-1 \right\}$.
 \end{example} 

\medskip
Note that the tensor product of abstract error groups is again an abstract
error group.
Now, assume that each party $x\in {\cal X}$ of the 
composite quantum system can be associated with an abstract error group
$W_x < \mathcal{U}(\cH_x)$ of unitaries
with center  $C_x$, which satisfy $W_x \supset C_x \subset  {\bf C}\1$,
such that $\wh{W_x} = W_x/C_x$ is abelian and has cardinality $d_x^2 \equiv |\cH_x|^2$.
Let $W \equiv \bigotimes_{x\in{\cal X}} W_x$
be the tensor product abstract error group,
acting on $\cH = \bigotimes_{x\in{\cal X}} \cH_x$.
For any subgroup  $\Gamma < W$, we let  
$\wh{\Gamma} =  (C\Gamma)/C \simeq \Gamma/(\Gamma \cap C)$ denote 
the quotient of $\Gamma$ by the  center of $W$.

Stabiliser codes~\cite{Gott:PhD,CRSS} are subspaces of $\cH$ which are simultaneous 
eigenspaces of abelian subgroups of $W$.

Consider a maximal abelian subgroup $G < W$, 
which contains the center $C = \bigotimes_{x\in{\cal X}} C_x <  {\bf C}\1$ of $W$,
so that $\wh{G}=G/C$ has cardinality $\sqrt{| \wh{W} |} = |\cH| = \prod_{j=1}^N |\cH_j|$.
Since $G$ is abelian it has a common eigenbasis, each state of which is 
called a stabiliser state $\ket{\psi}$.  

More generally, let $G < W$ be any abelian subgroup of an abstract error
group $W < {\cal U}(\cH)$.
Since all $g\in G$ commute, they
are jointly diagonalisable: let $P$ be one of the maximal
joint eigenspace projections.  Then  for $g\in G$,
$gP = \chi(g)P$, 
for a complex number $\chi(g)$. Thus $\chi:G\longrightarrow {\bf C}$
is necessarily the character of a $1$-dimensional group representation, 
which gives rise to the following expression for $P$:
\beqa   \label{eq:project1}
  P = \frac{1}{|G|} \sum_{g\in G}  \, \overline{\chi(g)} \,  g.
\eeqa
If $\chi(g_0) = 1$ and $g  = c \, g_0 $ is in the coset $g_0 \, C$, then 
$c = \chi(g )$ and $ \overline{\chi(g )} \,g  = g_0$.
Thus, $G_0 = \{ g \in G : \chi(g) = 1 \} $ is a subgroup of $G$ 
isomorphic to $\wh{G} = G/C$ and \eqref{eq:project1} can
be rewritten as
\beqa     \label{eq:project2}
    P =  \frac{1}{|G_0|} \sum_{g\in G_0}  g .
\eeqa     
Since $g^{-1} = g^\dag$ this sum is self-adjoint, and
\beqa
   P^2 =  \frac{1}{|G_0|^2} \sum_{g, h \in G_0}  g h =  \frac{1}{|G_0|} \sum_{g\in G_0}  g = P,   \nonumber
\eeqa
so that \eqref{eq:project2} is indeed a projection.

\medskip\noindent
{\bf Note:} The above reasoning is true because we assumed that
$\chi(g)$ records the eigenvalues of $g$ on the eigenspace with
projector $P$; as such, it has the property $\chi(t\1)=t$ for $t\in {\bf C}$.
For a general character $\chi$, however, only
$G_0 < \chi^{-1}(1)$ holds.

Because of the importance of the case of rank one projections, 
we summarize the results above in the case of maximal
abelian subgroups.
\begin{theorem}  \label{thm:stabstaterep}
Let $G$ be a maximal abelian subgroup of an abstract error
group $W$ with center $C$.   Any simultaneous eigenstate of
$G$ can be associated with a subgroup $G_0 \simeq G/C$
for which
  $~\ds{  \proj{\psi}  = \frac{1}{|G_0|} \sum_{g \in G_0} g}$.
\qed
\end{theorem}


\begin{remark}\normalfont
The use of the trivial representation is not essential in
the expression above.
It was used only to define $G_0$.   Once this has been
done, one can use the (1-dim) irreducible representations
of $G_0$ to describe the orthonormal basis of stabiliser
states associated with $G$.   Let  $\chi_k(g) $ denote the
$d = |G_0| $ irreducible representations of $G_0$ and
define
\beqa   \label{eq:stab-irred}
     \proj{ \psi_k } = \tfrac{1}{|G_0|} \sum_{g \in G_0}  \chi_k(g) g.
\eeqa
Then the orthogonality property of group characters implies
that 
$  \tr \proj{ \psi_j } \proj{ \psi_k } = | \langle \psi_j | \psi_k \rangle |^2 = \delta_{jk}  $.
\end{remark}


\subsection{Entropies of stabiliser states}    \label{sect:stab.ent}

The next result seems to have been obtained independently by several 
groups  \cite{HEB,Cubitt-etal,DDD}.

\begin{proposition} 
  \label{prop:stab-ent}
  For a pure stabiliser state $\rho = \proj{\psi}$ with associated error
  group $G < W$, and any $J \subset {\cal X}$,  the entropy
 \beqa  
 S(J)= S(\rho_J)  =  \log\frac{ d_J}{|\wh{G_J}  |}.
 \eeqa
 Here, $d_J = \prod_{x\in J} d_x$ and
 $$G_J \equiv \left\{ \otimes_{x\in{\cal X}} \, g_x \in G : 
                                \forall x\not\in  J \ g_x = \1 \right\}  \subset G, $$
 and  $\wh{G_J} = G_J/C_J $ is the quotient of $G_J$ with respect to the center $C_J = G_J \cap C$.
\end{proposition}
\begin{proof}
  It is enough to consider a bipartite  system with local
error groups $W_A$ and $W_B$, by considering party $A$ all
  systems in $J$, and $B$ all systems in ${\cal X}\setminus J$.
  Then,
  \[
    \proj{\psi} = \frac{1}{|\wh{G}|} \sum_{(g_A,g_B) \in \wh{G}} g_A \otimes g_B.
  \]
Since   $\tr g_B = 0$ unless  $g_B = \1$
  and   $ |\wh{G}| = d_A d_B $, this implies
  \[\begin{split}
    \rho_A &=     \trp_B \, \proj{\psi} = \frac{1}{|\wh{G}|} \sum_{(g_A,g_B) \in \wh{G}} (\tr g_B) \, g_A \\
           &= \frac{1}{|\wh{G}|} \sum_{g_A \in \wh{G_A}} |\cH_B|  \, g_A \\
           &= \frac{1}{|\cH_A|} \sum_{g_A \in \wh{G_A}} g_A
            = \frac{|\wh{G_A}|}{d_A}  \bigg( \frac{1}{|\wh{G_A}|} \sum_{g_A \in \wh{G_A}} g_A \bigg).
  \end{split}\]
   Since, $ \tr \rho_A = 1 $ , the last line implies that $\rho_A$ is  proportional to a
  projector of rank $\tfrac{d_A} {|\wh{G_A}|}$.
  Thus, its entropy is simply
  $  S(\rho_A) =  \log \frac{d_A}{|\wh{G_A} | } $.  
  \end{proof}
 
The following corollary is the key to our main result, Theorem~\ref{thm:stab-ingleton}.
\begin{corollary} 
  \label{cor:stab-S}
  For a pure stabiliser state as in Proposition~\ref{prop:stab-ent},
  the entropy of the reduced state $\rho_J$ satisfies
  \beqa
   S(J)= S(\rho_J) = \log \dfrac{| \wh{G} |  }{|  \wh{G_{J^c} }| } - \log d_J
                   = \log \dfrac{| \wh{G} |  }{|  \wh{G_{J} }| } - \log d_{J^c}.
  \eeqa
\end{corollary}
\begin{proof}
  As in Proposition~\ref{prop:stab-ent}, it suffices to consider the bipartite
  case.   Since    $\proj{\psi}$ is pure, 
\beeqa
   S(\rho_A) =  S(\rho_B) = \log \frac{d_B}{| \wh{G_B }|}  
           =  \log\frac{d_A d_B}{| \wh{G_B } |} - \log d_A .
\eeeqa
Since $d_A d_B =   |\wh{G} | $ this gives the desired result. 
\end{proof}

\begin{theorem}
  \label{thm:stab-ingleton}
  Any pure stabiliser state $\rho = \proj{\psi}$
  on an $5$-party system gives rise to $4$-party reduced states whose entropies
  satisfy the Ingleton inequality.    
\end{theorem}
\begin{proof}
By Corollary~\ref{cor:stab-S}, we have
\begin{equation}
  \label{eq:stab-S-rewrite}
  S(J)= \log\frac{|\wh{G}|}{| \wh{G_{J^c} }|} - \sum_{x\in J} \log d_x  \, .
\end{equation}
The first term $H(J)  = \log\frac{|\wh{G}|}{| \wh{G_{J^c} }|}$ is a Shannon entropy
of the type used in~\cite{ChanYeung-group}.   To be precise, 
observe that  $ \wh{G_{J^c} } = \bigcap_{x\in J} \wh{G_{\cX \backslash x} }$.
Moreover, since $\wh{G}$ and its subgroups $\wh{G_{J^c} }$ are abelian,
this implies that the entropy vector for each of the 4-party reduced states
satisfies  the Ingleton inequality. (This was observed in \cite{ChanYeung-group} 
and also follows from Theorem~\ref{thm:normsubgp}.)

To complete the argument, it suffices to observe that the Ingleton
inequality is balanced, so that the Ingleton expression is
identically zero for the sum-type ``rank function'' from the
second term in  \eqref{eq:stab-S-rewrite}, \emph{i.e.}
$h_0(J) \equiv  \sum_{x \in J} \log d_x $ defines a poly-matroid satisfying
(\ref{eq:ingleton}) with equality.
\end{proof}

Any linear combination of mutual informations and conditional mutual informations
is a balanced expression (and vice versa, any balanced expression can be written as
such a linear combination). Kinser's family of inequalities is balanced, which can
be seen by inspection of \eqref{eq:kinser}.    
It also   holds by construction
for the   inequalities obtained from
\cite[Thms.~3 and 4]{DFZ0910} and, more generally, any inequality obtained
using a ``common information'' as in  \cite{DFZ0910}.
Therefore, we can conclude using the same argument as above that
\begin{theorem}  
   Any pure stabiliser state on an $(N+1)$-party
   system generate an $N$-party entropy vector which
   satisfies the Kinser  \cite{Kinser}  family \eqref{eq:kinser} of inequalities,
   and more generally those of Dougherty et al.~\cite{DFZ0910}.
\end{theorem}
A consequence of Theorem~\ref{thm:stab-ingleton}
is that the \mat family of inequalities
holds for stabiliser states; however, rays generated by the
stabiliser state entropy vectors do not span the entropy cone 
$\ovb{\Sigma}^Q_4$.
In fact, from the proof of Theorem~\ref{thm:stab-ingleton}, we
see that \emph{every} balanced inequality that holds for the Shannon entropy,
holds automatically for stabilizer quantum entropies.\footnote{We 
are grateful to D. Gross and M. Walter, whose paper~\cite{GrossWalter} made us 
aware of this observation.}
Note also that apart from \eqref{mono}, all other necessary entropy inequalities for
the Shannon entropy are balanced~\cite{Chan:balanced}.


\section{The 4-party quantum entropy cone}  \label{sect:4party}

By direct calculation using symbolic software, we can
compute the extreme rays of $4$-party poly-quantoids plus Ingleton inequalities.
The results are given (up to permutation) as rays 0 to 6 in Table~1 below,
as elements of the 5-party  cone $\wtd{\Gamma}_{4+1}^Q$
(on subsets of $\{a,b,c,d,e\}$) of vectors which satisfy
($+$) \eqref{eq-ssa} and the complementarity
property  $S(J)= S(J^c)$ as described at the end of Section~\ref{sect:pure}.

\begin{table}[h]
\begin{center}
\begin{tabular}{l||c|c|c|c|c|c||c}
\raisebox{-2pt}{subset}$\big\backslash$\raisebox{2pt}{ray}
                         & 1 & 2 & 3 & 4 & 5 & 6 & 0 \\ \hline\hline
\ a                      & 1 & 1 & 1 & 1 & 2 & 1 & 1 \\ \hline
\ b                      & 1 & 1 & 1 & 1 & 1 & 1 & 1 \\ \hline
\ c                      & 0 & 1 & 1 & 1 & 1 & 2 & 1 \\ \hline
\ d                      & 0 & 1 & 1 & 1 & 1 & 2 & 2 \\ \hline
\ e ($\widehat{=}$ abcd) & 0 & 0 & 0 & 1 & 1 & 2 & 2 \\ \hline
ab                       & 0 & 1 & 2 & 1 & 3 & 2 & 2 \\ \hline
ac                       & 1 & 1 & 2 & 1 & 3 & 3 & 2 \\ \hline
ad                       & 1 & 1 & 2 & 1 & 3 & 3 & 2 \\ \hline
ae ($\widehat{=}$ bcd)   & 1 & 1 & 1 & 1 & 3 & 3 & 2 \\ \hline
bc                       & 1 & 1 & 2 & 1 & 2 & 3 & 2 \\ \hline
bd                       & 1 & 1 & 2 & 1 & 2 & 3 & 2 \\ \hline
be ($\widehat{=}$ acd)   & 1 & 1 & 1 & 1 & 2 & 3 & 2 \\ \hline
cd                       & 0 & 1 & 2 & 1 & 2 & 2 & 2 \\ \hline
ce ($\widehat{=}$ abd)   & 0 & 1 & 1 & 1 & 2 & 2 & 2 \\ \hline
de ($\widehat{=}$ abc)   & 0 & 1 & 1 & 1 & 2 & 2 & 2
\end{tabular}
\smallskip
\caption{Extreme rays of the 4-party quantum Ingleton cone}
\end{center}
\end{table}

The following stabiliser states found by Ibinson \cite{Ibinson-PhD}
(some of which were known earlier)
realise entropy vectors on the rays 1 through 6 shown in Table~1.
\begin{align}
  \label{eq:R1}
  \ket{\psi_1} &= \frac{1}{\sqrt{2}}\bigl(\ket{00}+\ket{11}\bigr)_{ab} \ket{000}_{cde}, \tag{R1} \\
  \label{eq:R2}
  \ket{\psi_2} &= \frac{1}{\sqrt{2}}\bigl(\ket{0000}+\ket{1111}\bigr)_{abcd} \ket{0}_e, \tag{R2} \\
  \label{eq:R3}
  \ket{\psi_3} &= \frac{1}{3} \sum_{i,j=0,1,2} \ket{i}_a \ket{j}_b \ket{i\oplus j}_c
                                                          \ket{i\oplus 2j}_d \ket{0}_e, \tag{R3} \\
  \label{eq:R4}
  \ket{\psi_4} &= \frac{1}{\sqrt{2}}\bigl(\ket{00000}+\ket{11111}\bigr)_{abcde},        \tag{R4} \\
  \label{eq:R5}
  \ket{\psi_5} &= \frac{1}{\sqrt{2}}\bigl(\ket{0}_{a'}\ket{0_L}_{a''bcde}
                                           +\ket{1}_{a'}\ket{1_L}_{a''bcde}\bigr),      \tag{R5} \\
  \label{eq:R6}
  \begin{split}
  \ket{\psi_6} &= \frac{1}{\sqrt{27}} \sum_{i,j,k=0,1,2} \ket{i}_a \ket{j}_b
                                                         \ket{i\oplus j}_{c'}\ket{k}_{c''}                                                           \ket{i\oplus j}_{d'}\ket{k}_{d''}
                                                         \ket{i\oplus j}_{e'}\ket{k}_{e''},
  \end{split}                                                                           \tag{R6}
\end{align}
where in eq.~(\ref{eq:R5}), $\ket{0_L}$ and $\ket{1_L}$ are the logical $0$
and $1$ on the famous $5$-qubit code~\cite{Laflamme:5qubit,BDSW}.
These are also extremal rays of the quantum entropy cone $\Sigma^Q_4$. In addition,
ray 0 in Table~1 is realised by the (stabiliser!) state
\begin{equation}
  \label{eq:state-0}
  \ket{\psi_0}
    = \frac{1}{2}\sum_{i,j=0,1} \ket{i}_A\ket{j}_B\ket{i\oplus j}_C\ket{ij}_D\ket{ij}_E. \tag{R0}
\end{equation}
on $1+1+1+2+2$ qubits.

Let us call an $N$-party poly-quantoid \emph{stabiliser-represented}, if
it is in the closure of the cone generated by the entropy
vectors of  $(N+1)$-party stabiliser states in the sense used above.
Then the above reasoning proves the following analogue of a theorem by
Hammer, Romashchenko, Shen and Vereshchagin~\cite{HRSV}:
\begin{theorem}
  \label{thm:stabiliser-quantoids}
  A $4$-party poly-quantoid is stabiliser-represented if and only if
  it satisfies the Ingleton inequality (and all its permutations).
  \qed
\end{theorem}

It seems reasonable to conjecture that the closure of the cone 
generated by the entropy vectors of stabiliser states is identical to that 
obtained when inequalities obtained from common information
as in \cite{DFZ0910} are added to the classical ones. However, it
is not even clear if stabiliser states satisfy the additional linear rank
inequalities shown to exist in \cite{CGK2}.

%

\section{Conclusion}

The difficult question of whether or not the quantum entropy
satisfies inequalities beyond positivity and SSA remains open
for four or more parties.

Do quantum states which do not satisfy Ingleton always
lie within the classical part of the quantum entropy cone?
We know that the quantum entropy cone $\ovb{\Sigma}_\cX^Q$
is strictly larger than the classical one $\ovb{\Sigma}_\cX^C$.
Recall that $\Lambda_4^{C,Q}$  denotes the
polyhedral cones formed from the classical inequalities (in
each case) and the Ingleton inequality.
We want to know whether or not   $\ovb{\Sigma}_4^Q \setminus  \Lambda_4^Q$ 
is strictly larger than  $\ovb{\Sigma}_4^C \setminus  \Lambda_4^C$, \emph{i.e.},  
are there quantum states whose entropy vectors do not satisfy the
Ingleton inequality and are not equal to any vector in the closure
of the classical entropy cone, $\ovb{\Sigma}_4^C$?  If the
answer is negative, then $4$-party quantum entropy vectors
must also satisfy the new non-Shannon inequalities. 

It seems that a better
understanding of quantum states which do not satisfy \eqref{eq:ingleton}
may be the key to determining whether or not quantum states 
satisfy the classical non-Shannon inequalities.  

This question extends naturally to the 5-party case, in which all
linear rank inequalities are known from  \cite{DFZ0910}.
However, for more parties,
one can ask the same question for both the cones associated
with inequalities obtained using one common information as in  \cite{DFZ0910},
and for the cones obtained using all linear rank inequalities.
Although we know from \cite{CGK,CGK2} that additional
inequalities are required, we do not even have explicit examples to consider.

\subparagraph*{Related work} 
After completion of the present research, we became aware
of independent work by Gross and Walter \cite{GrossWalter}, who use discrete
phase space methods for stabilizer states to show that the entropies of stabilizer 
states satisfy all balanced classical entropy inequalities. 
Indeed, this can also be seen
from our formula for the reduced state entropies in Corollary~\ref{cor:stab-S}.

\subparagraph*{Acknowledgements}
Portions of this work were done when FM, MBR and AW were 
participating in workshops at the IMS, National University of Singapore,
and at the Center for Sciences ``Pedro Pascual'' in Benasque, Spain.
This collaboration was a direct consequence of their participation in
a workshop on matroids held at the Banff International Research Station
in 2009.

The work of FM was partially supported by Grant Agency of the
Czech Republic under Grant 13-20012S.
The work of MBR was partially supported by NSF grant CCF-1018401.
NL and AW acknowledge financial support by the European Commission
(STREP ``QCS'' and Integrated Project ``QESSENCE''), AW furthermore
support by the ERC
(Advanced Grant ``IRQUAT''), a Royal Society Wolfson Merit Award
and a Philip Leverhulme Prize. The Centre for Quantum Technologies
is funded by the Singapore Ministry of Education and the National
Research Foundation as part of the Research Centres of Excellence
programme.


\begin{thebibliography}{50}

\bibitem{AG} R. Ahlswede, P. G\'{a}cs,
``Spreading of sets in product spaces and hypercontraction
of the Markov operator''
{\em Ann. Prob.}  {\bf 4}:925-939 (1976).


  \bibitem{AK}    R. Ahlswede, J. K\"orner, ``On Common
Information and Related Characteristics of Correlated Information Sources'' 
manuscript (1975); published in \emph{General Theory of Information Transfer
and Combinatorics}, LNCS Vol. 4123, Springer Verlag, 2006, pp. 664-677.

  \bibitem{BDSW} C. H. Bennett, D. P. DiVincenzo, J. A. Smolin, W. K. Wootters,
    ``Mixed-state entanglement and quantum error correction'' 
    {\em Phys. Rev. A}
    {\bf 54}(5):3824-3851 (1996).

   \bibitem{CLW} J. Cadney, N. Linden, A. Winter,
    ``Infinitely many constrained inequalities for the von Neumann entropy''
         {\em IEEE Trans. Inf. Theory}  {\bf 58}, 3657 (2012). 
         {\tt arXiv[quant-ph]:1107.0624}.
         
    \bibitem{CRSS} A. R. Calderbank, E. M. Rains, P. W. Shor, N. J. A. Sloane,
    ``Quantum Error Correction Via Codes Over GF(4)''
    {\em IEEE Trans. Inf. Theory} {\bf 44}(4):1369-1387 (1998).
             
    \bibitem{Chan:balanced} T. H. Chan, ``Balanced Information Inequalities''
    {\em IEEE Trans. Inf. Theory} {\bf 49}(12):3261-3267 (2003).

    \bibitem{CGK} T.  Chan, A. Grant, D. Kern,
      ``Existence of new inequalities for representable poly-matroids''
      {\em Proc. ISIT 2010}, pp. 1364-1368 (2010).
      {\tt arXiv[quant-ph]:0907.5030}.
 
   \bibitem{CGK2} T.  Chan, A. Grant, D. Pfl\"uger,   
     ``Truncation technique for characterizing linear poly-matroids''
     {\em IEEE Trans. Inf. Theory}   {\bf 57}:6364-6378  (2011).
  
  \bibitem{ChanYeung-group} T. H. Chan, R. W. Yeung,
    ``On a Relation Between Information Inequalities and Group Theory''
    {\em IEEE Trans. Inf. Theory} {\bf 48}(7):1992-1995 (2002).
    
  \bibitem{Cubitt-etal} D. Fattal, T. S. Cubitt, Y. Yamamoto, S. Bravyi, I. L. Chuang,
    ``Entanglement in the stabiliser formalism''
    {\tt arXiv:quant-ph/0406168} (2004).
    
    \bibitem{DDD}  M. Van den Nest, J. Dehaene, B. De Moor, 
      ``Local invariants of stabiliser codes''
      {\em Phys. Rev A}  {\bf  70}:032323 (2004). 
      {\tt arXiv:quant-ph/0404106}.
 
  \bibitem{DFZ2006} R. Dougherty, C. Freiling, K. Zeger, 
    ``Six New Non-Shannon Information Inequalities'' Proc. ISIT 2006,
    pp. 233-236 (2006).
    
  \bibitem{DFZ2007}  R. Dougherty, C. Freiling, K. Zeger,     
   ``Networks, Matroids, and Non-Shannon Information Inequalities''
    {\em IEEE Trans. Inf. Theory} {\bf 53}(6):1949-1969 (2007).
    
  \bibitem{DFZ0910} R. Dougherty, C. Freiling, K. Zeger,
    ``Linear rank inequalities on five or more variables''
    {\tt arXiv[cs.IT]:0910.0284} (2009).
    
   \bibitem{DFZ2011} R. Dougherty, C. Freiling, K. Zeger,
     ``Non-Shannon Information Inequalities in Four Random Variables''
     {\tt arXiv[cs.IT]:1104.3602} (2011).
   
    \bibitem{GK} P. G\'{a}cs, J. K\"orner, ``Common
     information is far less than mutual information'' 
     {\em Problems of Contr. and Inf. Theory} {\bf 2}:149-162 (1973).
   
  \bibitem{Gott:PhD} D. Gottesman, \emph{Stabilizer Codes and Quantum Error Correction},
    PhD thesis, Caltech, 1997.


  \bibitem{GrossWalter} D. Gross, M. Walter,
    ``Stabilizer information inequalities from phase space distributions''
    {\tt arXiv[quant-ph]:1302.6990} (2012).
    
  \bibitem{HRSV} D. Hammer, A. Romashchenko, A. Shen, N. Vereshchagin,
    ``Inequalities for Shannon Entropy and Kolmogorov Complexity''
  {\em   J. Comp. Syst. Sciences }{\bf 60}(2):442-464 (2000).
    
%
   \bibitem{HEB}   M. Hein, J. Eisert, H.J. Briegel
     ``Multi-party entanglement in graph states''
     {\em Phys. Rev. A}  {\bf 69}:062311 (2004).
     {\tt arXiv:quant-ph/0307130}.

  \bibitem{Ibinson-PhD} B. Ibinson, \emph{Quantum Information and Entropy},
    PhD thesis, University of Bristol, 2006 (unpublished). Available at URL
    {\tt http://www.maths.bris.ac.uk/$\sim$csajw/BenIbinson.PhD-} {\tt thesis-final.pdf}.

  \bibitem{Ingleton} A. W. Ingleton, ``Representation of matroids'' in: 
    \emph{Combinatorial Mathematics and its Applications}, ed. D. J. A. Welsh,
    pp. 149-167 (Academic Press, 1971). 
    
  \bibitem{Kinser}  R. Kinser, ``New Inequalities for subspace arrangements''
    {\em J. Comb Theory A} {\bf 118}:152-161 (2011).

  \bibitem{KlappiRoetti:nice-error-1} A. Klappenecker, M. R\"otteler, ``Beyond 
    Stabilizer Codes I: Nice Error Bases'' {\em IEEE Trans. Inf. Theory} 
    {\bf 48}(8):2392-2395 (2002).

  \bibitem{KlappiRoetti:nice-error-2} A. Klappenecker, M. R\"otteler, ``Beyond 
    Stabilizer Codes II: Clifford Codes'' {\em IEEE Trans. Inf. Theory} 
    {\bf 48}(8):2396-2399 (2002).

  \bibitem{Knill:error-groups} E. Knill, ``Group Representations, Error Bases and
    Quantum Codes'' LANL report LAUR-96-2807;
    {\tt arXiv:quant-ph/9608049} (1996).
    
  \bibitem{Laflamme:5qubit} R. Laflamme, C. Miquel, J. P. Paz, W. H. Zurek,
    ``Perfect Quantum Error Correcting Code'' 
    {\em Phys. Rev. Lett.} {\bf 77}(1):198-201 (1996).
    
    \bibitem{LYC}  S-Y.R. Li, R. Yeung, N. Cai,
    ``Linear network coding''
    {\em IEEE Trans. Inf. Theory} {\bf 49}:371-381 (2003).

  \bibitem{LiebRuskai} E. H. Lieb, M. B. Ruskai, ``Proof of the strong subadditivity 
    of quantum-mechanical entropy'' 
    {\em J. Math. Phys.} {\bf 14}(12):1938-1941 (1973).

  \bibitem{Lieb-inequalities} E. H. Lieb, ``Some Convexity and Subadditivity
    Properties of Entropy'' Bull. Amer. Math. Soc. {\bf 81}(1):1-13 (1975).


  \bibitem{LindenWinter:new} N. Linden, A. Winter, ``A New Inequality for the von Neumann
    Entropy'' 
    {\em Comm. Math. Phys.} {\bf 259}:129-138 (2005).
    
   \bibitem{MS}  F. Mat\'{u}\v{s},  M. Studen\'{y},
  ``Conditional independences among four random variables I.''
   {\em Comb. Prob. Comp.} {\bf 4}, 269-278 (1995).

  \bibitem{M2} F.~Mat\'{u}\v{s},
  ``Conditional independences among four random variables II.''
   {\em Comb. Prob. Comp.} {\bf 4}:407-417 (1995); III
   {\bf 8}:269--276 (1999).

  \bibitem{Matus} F. Mat\'{u}\v{s}, ``Infinitely Many Information Inequalities''
    Proc. ISIT 2007, pp. 41-44 (2007).

  \bibitem{Mbd} F. Mat\'{u}\v{s},  ``Two constructions on limits of entropy functions''    
    {\em IEEE Trans. Inf. Theory} {\bf 53}:320-330 (2007).

  \bibitem{M.poly} F. Mat\'{u}\v{s},  ``Polymatroids and polyquantoids''    
    Proc. WUPES 2012 (eds.\ J.\ Vejnarov\'a and T.\ Kroupa),
    Mari\'ansk\'e L\'azn\v{e}, Prague, Czech Republic, pp. 126-136 (2012).
    
 \bibitem{MTH} W. Mao, M. Thill, B. Hassibi,
    ``On the Ingleton-Violating Finite Groups and Group Network Codes''
     {\tt arXiv[cs.IT]:1202.5599} (2012).
   
  \bibitem{Oxley} J. G. Oxley, \emph{Matroid Theory}, Oxford University Press,
    Oxford, 2006.


  \bibitem{Pippenger} N. Pippenger, ``The inequalities of quantum information theory''
    {\em IEEE Trans. Inf. Theory} {\bf 49}(4):773-789 (2003).
    
  \bibitem{Shannon} C. E. Shannon,
 ``A mathematical theory of communication''
 {\em Bell System Technical Journal}  {\bf 27}:379-423 \&{} 623-656 (1948).
%
    \bibitem{SO} R. Stancu, F. Oggier
    ``Finite Nilpotent and metacyclic groups never violate the Ingleton inequality''
    2012 International Symposium on Network Coding (NetCod), pp. 25-30.

  \bibitem{Suzuki:groups} M. Suzuki, \emph{Group Theory I}, Springer
    Verlag, Berlin New York, 1982.

  \bibitem{Yeung:framework} R. W. Yeung, ``A Framework for Linear Information Inequalities''
    {\em IEEE Trans. Inf. Theory} {\bf 43}(6):1924-1934 (1997).

  \bibitem{ZhangYeung-1} Z. Zhang, R. W. Yeung,
    ``A Non-Shannon-Type Conditional Inequality of Information Quantities''
    {\em IEEE Trans. Inf. Theory} {\bf 43}(6):1982-1986 (1997).
  
  \bibitem{ZhangYeung-2} Z. Zhang, R. W. Yeung,
    ``On Characterization of Entropy Function via Information Inequalities''
    {\em IEEE Trans. Inf. Theory} {\bf 44}(4):1440-1452 (1998).
  
\end{thebibliography}
\end{document}